\documentclass[11pt]{article}

\usepackage{fullpage}
\usepackage{epsfig}
\usepackage{amsfonts}
\usepackage{amssymb}
\usepackage{amstext}
\usepackage{amsmath}
\usepackage{xspace}
\usepackage{shadow}


\newtheorem{theorem}{Theorem}
\newtheorem{corollary}[theorem]{Corollary}

\newtheorem{lemma}[theorem]{Lemma}

\newtheorem{fact}[theorem]{Fact}
\newtheorem{definition}[theorem]{Definition}
\newtheorem{proposition}[theorem]{Proposition}
\newcommand\ignore[1]{}

\newenvironment{proof}{\noindent{\bf Proof:  }}{\hfill\rule{2mm}{2mm}\medskip}

\def\sse{\subseteq}
\def\ts{\textstyle}
\def\Vbar{\overline{V}}
\def\dbar{\overline{d}}
\def\Mbar{\overline{M}}

\def\R{\mathbb{R}}
\def\T{\mathbb{T}}
\def\Z{\mathbb{Z}}
\def\e{\varepsilon}
\def\eps{\e}
\def\conv{{\sf conv}}
\def\leps{\tau}

\newcounter{note}[section]

 \newcommand{\agnote}[1]{}
 \newcommand{\ktnote}[1]{}

\newcommand{\initOneLiners}{%
    \setlength{\itemsep}{0pt}
    \setlength{\parsep }{0pt}
    \setlength{\topsep }{0pt}
}
\newenvironment{OneLiners}[1][\ensuremath{\bullet}]
    {\begin{list}
        {#1}
        {\initOneLiners}}
    {\end{list}}

\newcommand{\comment}[1]{}

\begin{document}

\title{How to Complete a Doubling Metric}

\author{Anupam Gupta\thanks{
   Computer Science Department, Carnegie Mellon University, Pittsburgh,
   PA 15213.  This research
   was partly supported by the NSF CAREER award CCF-0448095, and by an Alfred P.
   Sloan Fellowship.}
 \and Kunal Talwar\thanks{Microsoft Research,
Silicon Valley Campus, 1065 La Avenida, Mountain View, CA 94043.} }

\begin{titlepage}
\def\thepage{}
\thispagestyle{empty}

\date{}
\maketitle
\begin{abstract}

  \bigskip In recent years, considerable advances have been made in the
  study of properties of metric spaces in terms of their doubling
  dimension.  This line of research has not only enhanced our
  understanding of finite metrics, but has also resulted in many
  algorithmic applications. However, we still do not understand the
  interaction between various graph-theoretic (topological) properties
  of graphs, and the doubling (geometric) properties of the shortest-path
  metrics induced by them. For instance, the following natural question
  suggests itself: \emph{given a finite doubling metric $(V,d)$, is
    there always an \underline{unweighted} graph $(V',E')$ with
    $V\subseteq V'$ such that the shortest path metric $d'$ on $V'$ is
    still doubling, and which agrees with $d$ on $V$.} This is often
  useful, given that unweighted graphs are often easier to reason about.

  \medskip\noindent A first hurdle to answering this question is that
  subdividing edges can increase the doubling dimension unboundedly, and
  it is not difficult to show that the answer to the above question is
  negative. However, surprisingly, allowing a $(1+\eps)$ distortion
  between $d$ and $d'$ enables us bypass this impossibility: we show
  that for any metric space $(V,d)$, there is an \emph{unweighted} graph
  $(V',E')$ with shortest-path metric $d':V'\times V' \to \R_{\geq 0}$
  such that
  \begin{itemize}
  \item for all $x,y \in V$, the distances $d(x,y) \leq d'(x,y) \leq
    (1+\eps) \cdot d(x,y)$, and
  \item the doubling dimension for $d'$ is not much more than that of
    $d$, where this change depends only on $\e$ and not on the size of
    the graph.
  \end{itemize}

  \medskip\noindent We show a similar result when both $(V,d)$ and
  $(V',E')$ are restricted to be trees: this gives a simple proof that
  doubling trees embed into constant dimensional Euclidean space with
  constant distortion. We also show that our results are tight in terms
  of the tradeoff between distortion and dimension blowup.

\end{abstract}

\end{titlepage}
\newpage

\section{Introduction}

The algorithmic study of finite metrics has become a central theme
in theoretical computer science in recent years. Of particular
interest has been the study of the geometry of metrics---embeddings
into Minkowski spaces have been the most obvious example,
accompanied by the study of notions of metric dimension which have
allowed us to partially quantify geometric properties that make
metrics tractable for several algorithmic problems.

Given these advances in our understanding of the geometric
properties of abstract metric spaces, it is worth remarking that our
comprehension of the \emph{topological} properties of metric
spaces---and of the relationship between topology and geometry has
lagged behind: we do not yet have a good comprehension of how the
structure of a graph interacts with the dimensionality of the
shortest-path metric induced by it. One such example shows up in a
paper~\cite{GKL03}, where a fairly simple algorithm is given for
low-distortion Euclidean embeddings of unweighted trees whose
shortest-path metric is doubling---however, extending the result to
embed \emph{weighted} trees (also with doubling shortest-path
metrics) requires significantly more work. This raises the natural
question: \emph{given a doubling tree metric $M = (V,d)$, is there
an
  unweighted tree $G = (V', E')$ whose shortest-path metric is also
  doubling, and contains $M$ as a submetric?} In fact, the situation is
even more embarrassing: we do not know the answer even if we drop the
requirement that $G$ be a tree, and look for any unweighted graph!

An immediate obstacle to answering these question is the observation
that subdividing the edges of a weighted tree to convert it into an
unweighted tree can increase the dimension unboundedly. For example,
take a star $K_{1,n}$, and set the length of the $i^{th}$ edge
$\{v_0, v_i\}$ to be $2^i$. It is easy to check that the metric
$d_G$ has constant doubling dimension; however, subdividing the
$i^{th}$ edge into $2^i$ parts to make it unit-weighted creates a
new graph with $n$ points at unit distance from each other, which
has a doubling dimension $\log n$ that is unbounded. On the positive
side, it is easy to show that this metric can be embedded into the
real line with distortion $2$ (e.g., the map $v_i \mapsto 2^i$),
which we can subdivide without altering the doubling dimension. In
this paper, we show that this positive result is not an aberration:
any tree metric can be represented as a submetric of an unweighted
tree metric which has almost the same doubling dimension. We show a
similar result for arbitrary graphs as well, and show that our
tradeoff between distortion and the dimension blowup is
asymptotically optimal.

\medskip{\bf Formal Definitions:} To define the problems we study, let
us define the \emph{convex closure} of a graph, which is an extension of
the notion of subdividing edges. Given a graph $G = (V,E)$ with edge
lengths $\ell : E \to \R_{\geq 0}$, assume that the names of the
vertices in $V$ belong to some total order $(V, \prec)$. Let $\Vbar_G$
be the uncountably infinite set of points $V \cup \{ e[x] \mid e \in E,
x \in (0, \ell(e))\}$ obtained by considering each edge as a continuous
segment of length $\ell(e)$. Let $M_G = (V, d_G)$ be the shortest-path
metric of the graph $G$: we can define a natural metric on the set
$\Vbar_G$ as
\begin{align*}
  \dbar_G( e[x], e'[y] ) = \min\{ & x + d(u,u') + y, ~~ x + d(u,v') +
  (\ell(e') - y), \\ &(\ell(e) - x) + d(u',v) + y, ~~ (\ell(e) - x) + d(u',v')
  + (\ell(e') - y) \},
\end{align*}
if $e = \{u,v\}$ (with $u \prec v$) and $e' = \{u',v'\}$ (with $u' \prec
v'$). We now define the \emph{convex closure} of the graph $G$ to be the
metric space $\conv(G) \doteq \Mbar_G = (\Vbar_G, \dbar_G)$. Note the
metric obtained by subdividing edges of $G$ is a sub-metric of the
convex closure of $G$, and hence it suffices to study the doubling
dimension of this convex closure $\conv(G)$.

\subsection{Our Results}

The example of $K_{1,n}$ with exponential edge weights shows that even
if the shortest-path metric $M_G$ of a graph $G$ is doubling, its convex
closure $\Mbar_G$ may not be doubling. The goal of this paper is to show
that despite this, there is a ``close-by'' graph $G'$ whose convex
closure $\Mbar_{G'}$ is indeed doubling. In particular, the main theorem
is the following:
\begin{theorem}[Main Theorem]
  \label{thm:main1}
  Given a graph $G = (V,E)$ with specified edge-lengths, we can
  efficiently find a graph $G' = (V, E')$ (also with non-negative
  edge-lengths) such that
  \begin{OneLiners}
  \item The distances in $G$ and $G'$ are within a multiplicative factor
    of $(1 + \e)$ of each other, and
  \item If $\dim(M_{G}) = k$, then $\dim(M_{G'}) = O(k)$, and
    $\dim(\conv(G')) = O(k \log \e^{-1})$.
  \end{OneLiners}
\end{theorem}

Since Theorem~\ref{thm:main1} does not give any guarantees about the
topology of the graph $G'$, we prove an analogous result about tree
metrics, with improved guarantees on the dimension:
\begin{theorem}
  \label{thm:main2}
  Given a tree $T = (V,E)$ with specified edge-lengths, we can
  efficiently find a tree $T' = (V', E')$ with $V \sse V'$ (and with
  non-negative edge-lengths) such that
  \begin{OneLiners}
  \item For $x,y \in V$, the distance between them in $T$ and $T'$ are
    within a multiplicative factor
    of $(1+\e)$ of each other, and
  \item If $\dim(M_{T}) = k$, then $\dim(M_{T'}) = O(k)$, and
    $\dim(\conv(T')) = O(k + \log\log \e^{-1})$.
  \end{OneLiners}
\end{theorem}
As a corollary of this result, we obtain an independent proof of the
following result about embeddings of doubling tree metrics into $\ell_p$
spaces:
\begin{corollary}[\cite{GKL03}]
  \label{cor:gkl-thm}
  Every (weighted) doubling tree metric embeds into $\ell_p$ with
  constant distortion and constant dimension.
\end{corollary}
(Another proof of this embedding result for doubling trees appears
in~\cite{LNP06}, using completely different techniques.)

In addition, we show that the tradeoff between the distortion and
the dimension of the convex closure shown in Theorem~\ref{thm:main2}
is asymptotically optimal:

\begin{theorem}
\label{thm:treelowerbound} There exists a tree metric $T=(V,E)$ with
$\dim(M_T)=O(1)$ such that for any tree metric $T'=(V',E')$ with $V
\subseteq V'$, the following holds. If $d_T(u,v) \leq d_{T'}(u,v)
\leq (1+\e)d_T(u,v)$ for all $u,v \in V$, then $\dim(\conv(T'))$
must be $\Omega(\log\log\e^{-1})$.
\end{theorem}

For general graphs, we show that our tradeoff is asymptotically
optimal, under the restriction that the graph $G'$ is defined on the
same vertex set as $G$, i.e. we do not use any steiner points.
\begin{theorem}
\label{thm:graphlowerbound} There exists a metric $G=(V,E)$ with
$\dim(M_G)=O(1)$ such that for any graph $G'=(V,E')$, the following
holds. If $d_G(u,v) \leq d_{G'}(u,v) \leq (1+\e)d_G(u,v)$ for all
$u,v \in V$, then $\dim(\conv(G'))$ must be $\Omega(\log\e^{-1})$.
\end{theorem}

\textbf{Bibliographic Note.} James Lee informs us that a weaker form of
Theorem~\ref{thm:main1} can be inferred from results in a paper of
Semmes~\cite{Semmes96}. Indeed, the techniques of that paper imply that
for every doubling metric $G$, there is a graph $G'$ whose convex
closure is also doubling; however, the resulting distortion between
distances in $G$ and $G'$ using this approach seems to depend on
$\dim(G)$, and it is not clear how to reduce this distortion to $(1 +
\e)$ as in the results above.

\subsection{Related Work}
\label{sec:related-work}

The notion of doubling dimension was introduced by
Assouad~\cite{Assouad83} and first used in algorithm design by
Clarkson~\cite{Cla99}.  The properties of doubling metrics and their
algorithmic applications have since been studied extensively, a few
examples of which appear
in~\cite{GKL03,KL03,KL04,Tal04,HPM05,BKL06,CG06,IndykNaor,KRX06,KRX07}.

Somewhat similar in spirit to our work is the $0$-extension
problem~\cite{Karzanov98,CKR01-zero,FHRT03}. Given a graph $G$, the
0-extension ({\em cf.} Lipschitz
Extendability~\cite{JLS86,Mat90a,LeeN04}) problem deals with
extending a (Euclidean) embedding of the vertices of the graph to an
embedding of the convex closure of the graph, while approximately
preserving the Lipschitz constant of the embedding. Our results can
be interpreted as analogues to the above where the goal is to
approximately preserve the doubling dimension.

A number of papers have dealt with geometric implications of
topological properties of the graph inducing the metric, e.g. when
the graph is planar~\cite{KPR93,Rao99}, outer-planar~\cite{GNRS99},
series-parallel~\cite{GNRS99}, or a tree~\cite{Mat98}.

\section{Preliminaries and Notation}
\label{sec:prelims}

Given a graph $G$, the shortest path metric on it is denoted by $d_G$
and we shall use $B_G(x,r)$ to denote the ``ball'' $\{y \in V_G:
d_G(x,y) < r\}$. We will often omit the subscript $G$ when it is obvious
from context. There are several ways of defining the doubling constant
$\lambda$ and the doubling dimension $\dim$ for a metric space, all of
them within a constant factor of each other: here is the one that will
be most useful for us.
\begin{definition}[Doubling  Constant and Doubling Dimension]
  \label{def:doubling}
  A metric space $(X,d)$ has \emph{doubling constant} $\lambda$ if for
  each $x \in X$ and $r \geq 0$, given the ball $B(x,2r)$, there is a
  set $S \sse X$ of size at most $\lambda$ such that $B(x,2r) \sse
  \cup_{y \in S} B(y,r)$. The \emph{doubling dimension} $\dim((X,d)) =
  \log_2 \lambda$.
\end{definition}

\begin{fact}[Subset Closed]
  \label{fct:submetric}
  Let metric $M = (V,d)$ have doubling dimension $k$. If $X' \sse X$,
  and $d' = d|_{X' \times X'}$, then $(X', d')$ has doubling dimension
  at most $k$.
\end{fact}

\begin{fact}[Small Uniform Metrics]
  \label{fct:unif-metrics}
  If a metric $M = (V, d)$ has doubling dimension $k$ then there exists
  a point $x$ and a radius $r$ such that the ball $B(x,r)$ contains at
  least $2^k$ points with interpoint distances at least $r/2$.
\end{fact}

Given a metric $(X,d)$, an \emph{$r$-packing} is a subset $P \sse X$
such that any two points in $P$ are at least distance $r$ from each
other. An \emph{$r$-covering} is a subset $C \sse X$ such that for each
point $x \in X$, there is a point $c \in C$ at distance $d(x,c) \leq r$.
An \emph{$r$-net} is a subset $N \sse X$ that is both an $r$-packing and
an $r$-covering.
\begin{fact}[``Small'' Nets]
  \label{fct:net-size}
  Let metric $M = (V,d)$ have doubling dimension $k$, and $N$ is an
  $r$-net of $M$, then for any $x \in V$ and radius $R$, the set $B(x,R)
  \cap N$ has size at most $(4R/r)^k$.
\end{fact}

\begin{definition}[Geodesic Metrics] A metric
$(X,d)$ is said to be {\em geodesic} if for every $u,v \in X$, $u
\neq v$, there is a continuous map $f_{uv}: [0,d(u,v)] \rightarrow
X$ such that $f_{uv}(0)=u$, $f_{uv}(d(u,v))=v$ and
$d(f_{uv}(x),f_{uv}(y)) = |x-y|$ for any $x,y \in [0,d(u,v)]$.
\end{definition}

\begin{fact}
For any graph $G=(V,E)$, the metric $\conv(G)$ is a geodesic metric
space.
\end{fact}
\section{A Structure Theorem}
\label{sec:structure-theorem}

In this section, we show how to characterize the dimension of the
convex closure of a graph $H$ in terms of some easier-to-handle
parameters of the graph.
\begin{definition}[Long Edges]
  \label{def:long}
  Given a graph $H = (V,E)$, a vertex $u \in V$ and a radius $r \geq 0$,
  call an edge $e = \{v,w\}$ a \emph{long edge} with respect to $u,r$ if
  one endpoint of $e$ is at distance at most $r$ from $u$, and $l(e) > r$.
\end{definition}

Let the set of long edges with respect to $u,r$ be denoted by
$L_u(r)$. The following structure theorem gives us a
characterization of the doubling dimension of $\conv(H)$ in terms of
the number of long edges.
\begin{theorem}[Structure Theorem]
  \label{thm:struct}
  There exist constants $c_1$ and $c_2$ such that the following holds.
  Consider any graph $H = (V,E)$, and any $k \geq \dim_H$: if the
  number of long edges $|L_u(r)| \leq 2^{k}$ for every $u \in V$ and
  every $r \geq 0$, then the doubling dimension of the convex closure
  $\conv(H)$ is at most $c_1 k$. Moreover, if the
  doubling dimension of the convex closure $\conv(H)$ is at most $k$,
  then for every vertex $u \in V$ and every radius $r \geq 0$, the
  number of long edges $|L_u(r)| \leq 2^{c_2 k}$.
\end{theorem}

\begin{proof} Suppose the number of long edges $|L(u,r)|$ is at most $2^k$ for every $u,r$, then we
  show that for any $u\in \conv(H)$ and any $r>0$, the ball
  $B_{\conv(H)}(u,2r)$ can be covered by at most $2^{O(k)}$ balls
  $B_{\conv(H)}(y,\frac{3}{2}r)$. Repeating this argument three times
  (since $(\frac{3}{4})^3 < \frac{1}{2}$), this suffices to prove that
  the doubling dimension of $\conv(H)$ is $2^{O(k)}$.

  First, consider $u \in V(H)$. From the definition of doubling
  dimension, there is a set $S \subset V_H$ of size at most $2^{2\dim_G}
  \leq 2^{2k}$ such that $V_H \cap B(u,2r) \subseteq \cup_{y\in S} B(y,
  \frac{r}{2})$. Let $S' = S \cup \{e[r] \mid e \in L_u(r)\} \cup \{e[l(e)-r] \mid e \in L_u(r)\}$.  Clearly,
  $|S'| \leq 2^{O(k)}$. We shall show that $B(u,2r)$ is contained in
  $\cup_{y\in S'} B(y,\frac{3}{2}r)$.

  Let $e[x] \in B(u,2r)$, where $e = \{v,w\}$ such that $d(u,e[x]) =
  d(u,v)+x'$ where $x' \in \{x,l(e)-x\}$. Assume $v \prec w$ (the other case is similar) so that $x'=x$. If $x \leq r$, then consider $y
  \in S$ such that $v \in
  B(y,\frac{r}{2})$. Clearly $d(y,e[x]) \leq d(y,v)+x < \frac{3}{2}r$.
  Hence $e[x] \in B(y,\frac{3}{2}r)$. On the other hand, if $x > r$ and
  $e[x] \in B(u,2r)$, then $d(u,v) = d(u,e[x]) - x \leq r$ so that the
  edge $e$ is long with respect to $(u,r)$. Thus $e[r] \in S'$ and since
  $0\leq x \leq 2r$, we conclude that $e[x] \in B(e[r],r)$. Thus for any
  $u \in V_H$, $B(u,2r) \subseteq \cup_{y\in S'} B(y,\frac{3}{2}r)$.

  Finally, we have to consider balls around vertices in $\conv(H)
  \setminus V(H)$: note that for $e=\{u,v\}$, $B(e[x],2r) \subseteq
  B(u,2r) \cup B(v,2r) \cup \{e[z] \mid \max(0,x-2r) \leq z \leq
  \min(l(e),x+2r)\}$. By the argument above, the first two can be covered
  by a $2^{O(k)}$ balls of radius $r$ each. The subset of $e$ in
  $B(e[x],2r)$ is one dimensional and thus can be covered by two balls
  of radius $r$ each. This completes the argument showing that if the
  number of long edges is small, the convex completion has a small
  doubling dimension.

  For the converse, we shall show that $\dim(\conv(H)) \geq \Omega(\log
  \max_{u,r} L_u(r))$. Indeed consider the set of points $W=\{e[\frac{r}{2}] \mid
  r \in L_u(r)\}$. It is easy to see that $W \subseteq B(u,2r)$ but the
  balls $\{B(w,\frac{r}{2}) \mid w\in W\}$ are all disjoint. Thus $\dim(\conv(H))
  \geq \frac{1}{2}\log |W| = \frac{1}{2}\log |L_u(r)|$, whence the claim follows.
\end{proof}

The following simple result follows immediately.
\begin{corollary}
For any $n$ point metric $(X,d)$, there exists a geodesic metric
$(X',d')$ that contains an isometric copy of $(X,d)$ and has
doubling dimension at most $O(\log n)$.
\end{corollary}

The example of $K_{1,n}$ with exponential edge weights shows that
this bound is tight, even when $(X,d)$ itself has constant doubling
dimension. In the following sections, we show that when $(X,d)$
indeed has small doubling dimension, the $O(\log n)$ bound above can
be improved considerably if one allows a small distortion.

\section{Convex Completions for Graphs: Proof of  Main Theorem}
\label{sec:bd-spanner}

In this section, we show how to take a graph $G = (V,E)$ and obtain
a graph $G' = (V,E')$ on the same vertex set, which has (almost) the
same distances as in $G$, but whose doubling dimension does not
change by much under taking the convex closure. In particular, we
use a bounded-degree spanner construction due to Chan et
al.~\cite{CGMZ05}: they give an algorithm that given a metric
$(V,d)$ with dimension $\dim = \dim(G)$ and a parameter $\e < 1/4$,
outputs a spanner $G' = (V,E')$ such that $d(x,y) \leq d_{G'}(x,y)
\leq (1+\eps)\,d(x,y)$ for all pairs $x,y \in V$, and moreover the
degree of each vertex $x \in V$ is bounded by $\eps^{-O(\dim_G)}$.
We show that the convex closure of this spanner has doubling
dimension of $O(\dim_G \log \e^{-1})$.

\subsection{The Spanner Construction}
\label{sec:spanner-construction}

We start with a graph $G$ and carry out a series of transformations to
obtain graph $G'$. Let $\eps < \frac{1}{4}$ be given and let $\leps =
6+\lceil\log(\frac{1}{\eps})\rceil$. Without loss of generality, the
smallest pairwise distance in $G$ is at least $2^{\leps}$. We start with
some more definitions.
\begin{definition}[Hierarchical Tree]
  A hierarchical tree for a set $V$ is a pair $(\T,\phi)$, where $\T$ is
  a rooted tree, and $\phi$ is a labeling function $\phi : \T \rightarrow
  V$ that labels each node of $\T$ with an element in $V$, such that the
  following conditions hold.
  \begin{enumerate}
  \item Every leaf is at the same depth from the root.
  \item The function $\phi$ restricted to the leaves of $\T$ is a
    bijection into $V$.
  \item If $u$ is an internal node of $\T$, then there exists a child $v$
    of $u$ such that $\phi(v) = \phi(u)$. This implies that the nodes
    mapped by $\phi$ to any $x \in V$ form a connected subtree of $\T$.
  \end{enumerate}
\end{definition}

\begin{definition}[Net-Tree]
  \label{def:net-tree}
  A \emph{net tree} for a metric $(V,d)$ is a hierarchical tree
  $(\T,\phi)$ for the set $V$ such that the following conditions hold.
  \begin{enumerate}
  \item Let $N_i$ be the set of nodes of $\T$ that have height $i$. (The
    leaves have height $0$.) Let $r_0=1$, and $r_{i+1} = 2r_i$, for $i
    \geq 0$.  (Hence, $r_i = 2^i$.) Then, for $i \geq 0$,
    $\phi(N_{i+1})$ is an $r_{i+1}$-net for $\phi(N_i)$.
  \item Let node $u \in N_i$, and its parent node be $p_u$. Then,
    $d(\phi(u), \phi(p_u)) \leq r_{i+1}$.
  \end{enumerate}
\end{definition}
It is easy to see that net-trees exist for all metrics, and Har-Peled
and Mendel show how to construct a net-tree efficiently~\cite{HPM05}.

To construct their bounded-degree spanner, Chan et al.~\cite{CGMZ05}
define the following: suppose we are given a graph $G = (V,E)$, whose
shortest-path metric $(V,d_G)$ has doubling dimension $\dim_G$. Let
$\eps > 0$ and $(\T,\phi)$ be any net tree for $M$. For each $i > 0$, let
\begin{gather}
  E_i := \Big\{\{u,v\} \mid u,v \in \phi(N_i), d_G(u,v) \leq (4 +
  \frac{32}{\eps}) \cdot r_i \Big\} \setminus \bigcup_{j \leq i-1} E_{j} ,
\end{gather}
where $E_0$ is the empty set. (Here the parameters $N_i, r_i$ are as in
Definition ~\ref{def:net-tree}.)  Letting $C_{\eps}$ denote
$(4+\frac{32}{\eps})$, we note that all edges in $E_i$ have length in
$(C_{\eps}r_{i-1},C_{\eps}r_i]$.

While the graph $\widehat{G} = (V, \widehat{E} = \cup_i E_i)$ is a
$(1+\e)$-spanner for the original metric with few edges, obtaining a
bounded-degree spanner requires some modifications to the basic
construction. First, the edges in $\widehat{E}$ are directed (merely
for the purposes of the algorithm, and the proof). For each $v \in
V$ , define $i^*(v) := \max\{i| v \in \phi(N_i)\}$. For each edge
$(u,v)\in \widehat{E}$, we direct the edge from $u$ to $v$ if
$i^*(u) < i^*(v)$. If $i^*(u) = i^*(v)$, the edge can be directed
arbitrarily. Chan et al. show that each vertex $x \in V$ has
\emph{out-degree} bounded by $\beta = \eps^{-O(\dim_G)}$. Then, the
following steps are performed:
\begin{itemize}
\item Consider any vertex $x$, and all the edges that are directed
  \emph{into} $x$. These edges come from various sets $E_i$: let us
  denote by $F_i = F_i(x)$ the subset of edges directed into $x$ that
  belong to $E_i$.
\item Suppose the non-empty subsets are $F_{i_1}, F_{i_2}, \ldots,
  F_{i_t}$, where $i_j < i_{j+1}$. We do nothing to the first $7 \log \e^{-1}$ of these edge
  sets; these contribute $\eps^{-O(\dim_G)}$ to the final degree of~$x$.
\item Consider a value of $j > 7 \log \e^{-1}$: from the set $F_{i_{(j -
      7 \log \e^{-1})}}$ of edges directed \emph{into} $x$, we choose an
  arbitrary one $\{u,x\}$.  We replace edges of the form $\{y,x\} \in
  F_{i_j}$ by edges $\{y,u\}$---and refer to these (at most
  $\e^{-O(\dim_G)}$) edges as edges \emph{donated} from $x$ to
  $u$.

  Note that the length of the edge $\{u,x\}$ is at most $C_\e \,
  2^{i_{(j - 7 \log \e^{-1})}} \leq C_\e \, \e^72^i$, whereas the length of
  any edge in $\{y,x\} \in F_{i_j}$ is at least $C_\e\, 2^{i-1}$; hence
  $d_G(u,x) \leq (\e^7/2) d_G(x,y) \leq \e^6 d_G(x,y)$, since $\e \leq 1/4$.
  By the triangle inequality, $d_G(u,y) \in (1 \pm \e^6) d_G(x,y)$.

  Additionally, note that if $x$ donates a long edge $(x,y) \in
  F_{i_j}$ to $u$, then $(u,x) \in F_{i_{j - 7 \log \e^{-1}}}$ so that $d_G(x,u)$ is
  at least $\smash{C_{\e}2^{(i_{j - 7 \log \e^{-1}})-1}}$.
\end{itemize}

\begin{theorem}[\cite{CGMZ05}]
  \label{thm:cgmz}
  The spanner thus constructed has degree $\e^{-O(\dim_G)}$ and stretch
  $(1+\e)$.
\end{theorem}

From the construction of the bounded-degree spanner, note that each
vertex $u \in V$ has the following edges incident to it:
  \begin{itemize}
  \item \textbf{Type-A edges.} These correspond to the $\eps^{-O(\dim_G)}$
    edges that were directed \emph{away} from $u$.

  \item \textbf{Type-B edges.} These correspond to the edges directed
    \emph{into} $u$ that belong to the smallest $7\log \eps^{-1}$
    levels; this gives another $(\eps^{-O(\dim_G)})$ edges in total.

  \item \textbf{Type-C edges.} For each edge $e = \{u,x\}$ of type-A
    incident to $u$, there are at most $(\eps^{-O(\dim_G)})$ other edges
    incident to $u$ that are not counted above. Each such edge $e' =
    \{y,u\}$ corresponds to some edge of the form $\{y,x\} \in E_i$ (for
    some $i$ such that $x,y \in \phi(N_i)$), such that the edge was
    ``donated'' from $x$ to $u$ to maintain $x$'s degree bound.
  \end{itemize}

\subsection{Bounding the Dimension of the Convex Closure}
\label{sec:conv-claim}

Simply by the distortion bound, it follows that the doubling dimension
of the bounded-degree spanner $G'$ is close to $\dim_G$. Of course, the
bounded-degree does not imply that $\conv(G')$ has low doubling
dimension: in this section, we use the Structure
Theorem~\ref{thm:struct} to show this fact, and hence prove
Theorem~\ref{thm:main1}.

\begin{lemma}
  \label{lem:bd-spanner-longs}
  Given the graph $G'$ defined as above, fix any vertex $v$ and radius
  $R$, and $\e < \frac14$. Then the number of long edges $|L_v(R)|$ with
  respect to $v,R$ is at most $O(\e^{-O(\dim_G)})$.
\end{lemma}

\begin{proof}
  Recall that $L_v(R)$ is the set of edges that have one endpoint within
  the ball $B(v,R)$, and have length at least $R$.
  Define $\ell \in \Z_{\geq 0}$ such that $R \in (C_\e 2^{\ell-1}, C_\e
  2^{\ell}]$.

  By the spanner construction, any type-A or type-B edge that is long
  must belong to $\cup_{i \geq \ell}\, E_i$, and hence must have both
  endpoints in $\phi(N_\ell)$. Moreover, one endpoint of each such a
  long edge must lie in the ball $B(v,R) \sse B(v, C_\e 2^{\ell})$;
  since the points in $\phi(N_\ell)$ are at distance at least $2^\ell$
  from each other, there can be at most $(C_\e)^{O(\dim_G)}$ many such
  endpoints within the ball. Moreover, each one of these endpoints has
  at most $\eps^{-O(\dim_G)}$ type-A or type-B edges; multiplying them
  together, using the fact that $C_\e = O(\e^{-1})$, and simplifying
  gives an upper bound of $\eps^{-O(\dim_G)}$ on the number of type-A and
  type-B edges in $L_v(R)$.

  Let us now consider the edges in $L_v(R)$ that are of type-C with
  respect to their endpoint within $B(v,R)$. Recall that each type-C
  edge $\{u,y\}$ can be associated with some edge $\{x,y\} \in
  \widehat{E}$ (of almost the same length---up to a factor of $(1 \pm
  \e^6)$) such that $x$ donates the edge to $u$. Let us fix one such
  long edge $e = \{u,y\}$ associated with $\{x, y\} \in E_i$---hence the
  distance $d_G(x,y) \in (C_\e 2^{i-1}, C_\e 2^i]$, and also $x,y \in
  \phi(N_i)$.  By the construction of the type-C edges, the distance
  $d_G(u,x) \leq \e^6 \cdot d_G(x,y)$, and hence $x$ lies in the ball $B(v,
  R + \e^6 C_\e 2^i)$.

  Given any fixed level $i \geq \ell-1$, the number of donor vertices is
  bounded by the number of points in $B(v, C_\e (2^\ell + \e^6 2^i))$
  that are at least $2^i$ distance apart from each other, which can be
  loosely bounded by $\e^{-O(\dim_G)}$. Each such donor vertex could
  donate $\e^{-O(\dim_G)}$ edges, which would give us a total of
  $\e^{-O(\dim_G)}$ edges for the level $i$. Summing this over all levels
  would give us too many edges, so we use this bound only for levels $i$
  such that $\ell - 1 \leq i \leq \ell + O(\log \e^{-1})$.

  Consider any level $i > \ell + 6\log \e^{-1}$: any donor vertex for
  such a level must lie in the ball $B(v, C_\e ( 2^\ell + \e^6 2^i))
  \sse B(v, C_\e \e^6 2^{i+1}) \sse B(v, C_\e\, \e^5\,2^i)$. A little
  algebra shows that
  \[ \e^5 C_\e = \e^5 (4 + \frac{32}{\e}) \leq \e^5
  \frac{33}{\e} \leq \e^4 \cdot 33 \leq \e,\] and thus the donor vertex
  must be at distance at most $\e 2^i$ from $v$. However, since the
  donor vertex must belong to $\phi(N_i)$, it must be at distance at
  least $2^i$ from any other donor vertices. Now, if there were two
  donor vertices at distance $\e 2^i$ from $v$, they would be at
  distance $2\e 2^i < 2^i$ from each other---this implies that there can
  be at most one donor vertex for such a ``high'' level.

  Finally, it remains to show that the total number of long edges
  donated by this donor vertex $x$ to vertices in $B(v,R)$ is small.
  Let $i_1,i_2,\ldots,i_t$, $i_j < i_{j+1}$ be the levels for which $x$
  donates a long edge to vertices in $B(v,R)$; we shall show that
  $t$ is at most $O(\log \e^{-1})$. Since the first edge is long, $R \leq
  C_\e  2^{i_1+1}$. Moreover, since $x$ donates this edge to $u$, we
  conclude that $d_G(x,u_1) \leq \eps^6 C_\e 2^{i_1}$, so that $d_G(v,x)
  \leq R + \eps^6C_\e 2^{i_1} \leq (2+\eps^6)C_\e 2^{i_1}$.
  Suppose that $t > 7\log\e^{-1} + 3$. Then an edge in $F_{i_t}$
  is donated from $x$ to $u_t$, and we have that $d_G(x,u_t) \geq C_\e 2^{i_4 -
  1}$. On the other hand, since $u_t \in B(v,R)$, by triangle
  inequality, $d_G(x,u_t) \leq d_G(x,v)+ d_G(v,u_t) \leq (3+\eps^6)C_\e
  2^{i_1}$. Since $i_4 \geq i_1 + 3$, this gives us the desired
  contradiction. Thus $t \leq O(\log \e^{-1})$. Since there are at most
  $\e^{-O(\dim_G)}$ edges donated to $B(v,R)$ from each of these levels, the
  claim follows.
\end{proof}

Using Lemma~\ref{lem:bd-spanner-longs} along with the Structure
Theorem~\ref{thm:struct} implies that the dimension of $\conv(G')$ is
bounded by $O(\dim_G \log \e^{-1})$, which proves Theorem~\ref{thm:main1}.

\section{Convex Completions for Trees}
\label{sec:conv-compl-trees}

The construction of the previous section showed that given any graph
$G$, we could construct a new graph $G'$ such that distances in $G$ and
$G'$ are within $(1+\e)$ of each other, and $\conv(G')$ has low doubling
dimension. However, since the construction starts with the shortest-path
metric $d_G$ and completely ignores the topological structure of $G$
itself, it is not suited to proving Theorem~\ref{thm:main2} which seeks
to start with a tree and end with another tree. In this section, we show
a different approach that allows us to monitor the graph structure more
closely.

\subsection{The Construction for Trees}
\label{sec:construction-trees}

We give a procedure that takes a general graph $G$ and outputs a graph
$G'$ (since the construction itself does not depend on $G$ being a
tree); we then show some properties that hold when $G$ is a tree. The
procedure takes a graph $G = (V,E)$, and constructs a new graph $G' =
(V', E')$ with $V \sse V'$ (by way of an intermediate graph
$\widehat{G}$) as follows. Define an \emph{exponential tail} with $k$
edges as a path $P = \langle v_0, v_1, v_2, \ldots, v_k \rangle$, where
the length of the edge $\{v_{i-1}, v_{i}\}$ is $2^i$. Without loss of
generality, the smallest edge length in $G$ is at least $2^{\leps}$,
where $\leps = 6+\lceil\log(\frac{1}{\eps})\rceil$.

We construct the graph $G'$ in the following way:
\begin{itemize}
\item As in Section~\ref{sec:spanner-construction}, we consider a
  net-tree $(\T, \phi)$ for the graph $G$. If $N_i$ is the set of nodes
  in $T$ at height $i$, then for $u \in V$ define $i^*(u)$ to be the
  largest $i$ such that $u \in \phi(N_i)$. Attach to each $u \in V$ an
  exponential tail with $i^*(u)$ edges; refer to the $j^{th}$ vertex on
  this path as $u_{[j]}$, with $u_{[0]} = u$. Let $\widehat{G}$ be this
  intermediate graph consisting of $G$ along with the tails.

\item Consider an edge $e = \{u,v\} \in E(G)$, and suppose its length
  lies in the interval $(C_\e 2^{i-1}, C_\e 2^{i}]$. Some leaf of $T$
  must be mapped by $\phi$ to $u \in V$: let the level-$(i)$ ancestor of
  that node be mapped by $\phi$ to $\widehat{u}$; similarly, define
  $\widehat{v}$ be defined for $v$. We now make an edge $\{
  \widehat{u}_{[i]}, \widehat{v}_{[i]} \}$ of length $\ell_e$ in the
  graph $G'$.
\end{itemize}
Note that if we start off with a tree $T$, the above procedure adds
exponential tails to $T$ to get the intermediate graph $\widehat{T}$,
and then ``moves the edges up the tails'' to get the final graph $T'$.

\begin{proposition}[Distance Preservation]
  \label{fct:stretch-tree}
  Let $\e < 1/4$. If the input graph is a tree $T = (V,E)$, then the
  above procedure results in a connected tree $T' = (V',E')$ such that
  for any $x,y \in V$,
  \[ (1+\e)^{-1}d_T(x,y) \leq d_{T'}(x,y) \leq (1+\e)d_T(x,y). \]
\end{proposition}

\begin{proof}
  Let us consider performing the above-mentioned transformation for
  edges in increasing order of edge-length. Given $j \in \Z_{\geq 0}$,
  let $T_j$ be the forest formed by deleting all edges of length more
  than $C_\e 2^j$ from $T$; also, let $T'_j$ be the forest formed by
  deleting the corresponding edges in $T'$. We will prove by induction on $j$
  that for all $x,y$ that lie in some connected component in $T_j$,
  their distance in $T_j'$ will satisfy the desired stretch bound. The
  base case is trivial, since all components of $T_0$ have single nodes
  in them.

  To prove the claim for $j$, we inductively assume it for $j-1$. Now
  consider taking some edge $e = \{u,v\}$ of length $\ell_e \in (C_\e
  2^{j-1}, C_\e 2^{j}]$. In this case we find some nodes $\widehat{u}$
  and $\widehat{v}$, and add an edge of length $\ell_e$ between
  $\widehat{u}_{[j]}$ and $\widehat{v}_{[j]}$.  By the properties of the
  net-tree, the distance $d_T(u, \widehat{u}) \leq 2^{j+1} - 2$. Since
  $T_j$ already contains all edges of length at most $C_\e 2^{j-1}$, and
  $C_\e \geq 4$, the net point $\widehat{u}$ lies in the same component
  as $u$ in $T_j$. By the induction hypothesis, $d_{T'}(u,\widehat{u}) \leq
  (1+\e) 2^{j+1}$; note that this implicitly proves that $u$ and
  $\widehat{u}$ are in the same component in $T'_j$. A similar claim holds for
  $d_T(v, \widehat{v})$. Hence the distance in $T'_{j+1}$ between $u$
  and $v$ is at most
  \begin{align*}
    & d_{T'_j}(u, \widehat{u}) + d_{T'_j}(\widehat{u},
    \widehat{u}_{[j]}) + \ell_e + d_{T'_j}(\widehat{v}_{[j]},
    \widehat{v}) +
    d_{T'_j}(\widehat{v}, v) \\
    & = 2 \times (1+\e) 2^{j+1} + 2 \times 2^{j+1} + \ell_e \\
    & \ts \leq \ell_e\; (\frac{8(1+\e) + 8}{C_\e} + 1) \leq
    (1+\e)\ell_e,
  \end{align*}
  where we used the fact that $C_\e = (4 + \frac{32}{\e})$ and $\e <
  1/4$. Since each of the edges of $T$ are not stretched by more than
  $(1+\e)$, this implies that the stretch for all pairs is bounded by
  the same value.

  We also need to show that the distances are not shrunk too much in
  $T'$: to show this, we go via $\widehat{T}$. (Recall that
  $\widehat{T}$ was the original tree $T$ along with the exponential
  tails.) First note that for any $u,v \in V$,
  $d_{T}(u,v)=d_{\widehat{T}}(u,v)$. We show that distance do not shrink
  in going from $\widehat{T}$ to $T'$.  It suffices to show this for the
  edges of $T'$. For an edge $e'=(\widehat{u}_{[j]},\widehat{v}_{[j]})$
  that has length $\ell_e \geq C_\e2^{j-1}$, we note that their distance
  in $\widehat{T}$
  \begin{gather}
d_{\widehat{T}}(\widehat{u}_{[j]},\widehat{v}_{[j]}) \leq
d_{\widehat{T}}(\widehat{u}_{[j]}, \widehat{u}) +
d_{\widehat{T}}(\widehat{u},u) + \ell_e +
d_{\widehat{T}}(v,\widehat{v}) +
    d_{\widehat{T}}(\widehat{v}, \widehat{v}_{[j]}) \leq 4(2^{j+1}-2) + \ell_e
  \end{gather}
Since $C_\e > 32/\e$, this is at most $(1+\e)\ell_e$. Thus the
contraction going from $\widehat{T}$ to $T'$ is at most $(1+\e)$.

Finally, we note that we have shown that $T'$ is connected, and the
number of edges in $T'$ is equal to the number of edges in
$\widehat{T}$, which is a tree. Thus $T'$ is a tree as well.
\end{proof}

\subsection{Bounding the Dimension of the Convex Closure: The Tree Case}

Finally, to show that the doubling dimension of $\conv(T')$ is
small, we will again invoke Theorem~\ref{thm:struct}. However, since
we have added additional vertices in going from $T$ to $T'$, we
first show that $\dim(T')$ is $O(\dim(T))$. Since we have already
shown that distances are preserved in going from $\widehat{T}$ to
$T'$, it suffices to bound the doubling dimension of $\widehat{T}$.

\begin{lemma}
The doubling dimension of $\widehat{T}$ is at most $O(\dim(T))$.
\end{lemma}
\begin{proof}
Let $u_{[i]} \in V(\widehat T )$ and $R \geq 0$ with $R \in
(2^{j-1},2^j]$. We wish to show that $B(u_{[i]},2R)$ can be covered
by a small number of balls of radius $R$. From the definition of
doubling dimension, it follows that there is a set $Y$ with $|Y|
\leq 2^{2\dim(T)}$ such that $B_{T}(u,2R) \subseteq \cup_{y \in Y}
B_{T}(u, R/2)$. Note that for any $v \not \in \phi(N_{j-2})$, the
tail attached to $v$ has length at most $R/2$. Let $Z = B(u, 2R)
\cap \phi(N_{j-2})$; clearly $|Z| \leq 2^{O(\dim(T))}$. Finally, let
$Z' = \{v_{[j-1]} : v \in Z\}$ and $Z''=\{v_{[j]} : v \in Z\}$. It
is easy to verify that $B(u_{[i]},2R) \subseteq \cup_{y \in Y \cup Z
\cup Z'\cup Z''} B(y,R)$. The claim follows.
\end{proof}

Finally, we show the following bound on the number of long edges in
$T'$.
\begin{lemma}[Few Long Edges]
  \label{lem:long-edges-tree}
  For any vertex $v \in T'$ and every radius $R$, the number of long
  edges in $T'$ is bounded by $2^{O(\dim)}\log \e^{-1}$.
\end{lemma}

\begin{proof}
   First consider some $v \in V$, and $R \geq 0$, and define $\ell \in \Z_{\geq 0}$ such that $R \in (C_\e 2^{\ell-1}, C_\e
  2^{\ell}]$.
  Every long edge incident on $B(v,R)$ must have length at least
  $R$. Further, edges longer than $2 C_\e R$ are incident on a tail
  node further than $R$ from its root, and hence such an edge cannot
  be incident on $B(v,R)$. For each of the length scales
  $(C_\e 2^{\ell+j-1}, C_\e 2^{\ell+j}): 0 \leq j \leq \log C_{\e}$,
  we will bound the number of long edges in that length scale.
  Fix one such scale, and let $L(v,R,j)= \{(u_i,w_i): 1 \leq i \leq |L(v,R,j)|\}$ be the set
  of long edges of length in $(C_\e 2^{\ell+j-1}, C_\e 2^{\ell+j})$,
  such that $d(v,u_i) \leq R$. Since each long edge has length
  more than $R$, there is a path from $v$ to $u_i$ that does not use
  any of the long edges. Consider the set of nodes $W=\{w_i: 1 \leq i \leq |L(v,R,j)|\}$. Clearly,
  for any $w,w'\in W$, $d(w,w')$ is at most $2R + 2 C_\e 2^{\ell+j} \leq 4 C_\e 2^{\ell+j}$. Moreover, since $T$ is
  a tree, the symmetric difference of the $v$-$w$ and $v$-$w'$ paths
  gives the shortest path from $w'$ to $w$. Since the long edges
  incident on $w$ and $w'$ are in this symmetric difference, we
  conclude that $d(w,w') \geq 2 C_\e 2^{\ell+j-1}$. Thus from the
  bound on doubling dimension, we conclude that $|W| \leq
  2^{O(\dim)}$. Adding the contribution of the $O(\log \e^{-1})$
  distance scales, we get the desired bound.

  We now extend the argument to a vertex $v_{[i]}$ on an exponential
  tail hanging off $v$. If $i \geq j$, then $B(v_{[i]},R) = \{v_{[i]}\}$.
  All edges incident on $v$ have, up to a factor of two, the same
  length, and thus their endpoints form a near uniform submetric.
  Thus we can bound the degree of $v_{[i]}$ by $2^{O(\dim)}$
  and the claim follows. On the other hand, when $i < j$, $B(v_{[i]},R) \subseteq B(v,2R)$ and an
  argument analogous to the one for the case $v \in V$ above
  suffices.
\end{proof}

Theorem~\ref{thm:main2} follows.

\section{Lower Bounds}

In this section, we show that the tradeoff between distortion and
dimension blowup is asymptotically optimal. Consider the graph
$K_{1,n}$ with $v_0$ as the center node and $\{v_1,\ldots,v_n\}$ as
the set of leaves. Set the length of the edge $\{v_0,v_i\}$ to $2^i$
and let $d$ be the resulting metric on the vertices $V$ of
$K_{1,n}$. It is easy to check that this metric has constant
doubling dimension. We next show that the doubling dimension of any
geodesic metric $(X,d')$ containing a $(1+\e)$-distortion copy of
$(V,d)$ is $\Omega(\log\log\e^{-1})$.
\begin{lemma}Let $(X,d')$ be any geodesic metric such that $V
\subseteq X$ and $d(v_i,v_j) \leq d'(v_i,v_j) \leq
(1+\eps)d(v_i,v_j)$ for all $v_i,v_j \in V$. Then $\dim(X,d')$ is
$\Omega(\log\log\e^{-1})$.
\end{lemma}
\begin{proof}
Denote by $uw[x]$ the point on the shortest $u$-$w$ path in $X$ that
is at distance $x$ from $u$ (if there is more than one shortest
path, pick one arbitrarily). We shall argue that the points
$v_0v_i[1]$ for $i=\{1,\ldots,\log (2\e)^{-1}\}$ form a large
near-uniform submetric in $X$. Indeed $d'(v_0v_i[1],v_0v_j[1]) \leq
d'(v_0v_i[1],v_0)+d'(v_0,v_0v_j[1]) = 2$. On the other hand, by
triangle inequality,
\begin{eqnarray*}
d'(v_0v_i[1],v_0v_j[1]) &\geq& d'(v_i,v_j) - d'(v_0v_i[1],v_i) -
d'(v_0v_j[1],v_j)\\
&=& d'(v_i,v_j) - (d'(v_0,v_i)-1) - (d'(v_0,v_j)-1)\\
&\geq& 2 + d(v_i,v_j) - (1+\e)(d(v_0,v_i)+d(v_0,v_j))\\
&=& 2 - \e(2^i+2^j)\\
\end{eqnarray*}
where we have used the bound on the distortion and the distance
definitions in $d$ in the last two steps. Since $i,j \leq
\log(2\e)^{-1}$, we conclude that $d'(v_0v_i[1],v_0v_j[1]) \geq 1$.
Thus we have $\log (2\e)^{-1}$ points in $X$ that lie within
$B(v_0,2)$ no two of which can be covered by a single ball of radius
$\frac{1}{2}$. Thus the doubling dimension of $X$ is
$\Omega(\log\log\e^{-1})$.
\end{proof}

Theorem~\ref{thm:treelowerbound} follows.

For general metrics, we show a stronger lower bound, under a
stronger constraint on $X$. Let $V=\{0,1\}^p$ with $d(x,y) =
2^{p-lcp(x,y)}$, where $lcp(x,y)$ denotes the length of the longest
common prefix of strings $x$ and $y$. Once again, one can easily
check that $(V,d)$ has constant doubling dimension. We show that any
graph $H=(V,E)$ on $V$ approximating $d$ within distortion $(1+\e)$
must satisfy $\dim(\conv(H)) \in \Omega(\log \e^{-1})$.
\begin{lemma}
Let $H= (V,E)$ be any graph such that the shortest path metric $d'$
satisfies $d(x,y) \leq d'(x,y) \leq (1+\e)d(x,y)$ for all $x,y \in
V$. Then $\dim(\conv(H))$ is $\Omega(\log \e^{-1})$.
\end{lemma}
\begin{proof}
For $p=\log(2\e)^{-1}$, we first show that $H$ must have all edges
connecting $V_0= \{0x: x\in \{0,1\}^{p-1}\}$ and $V_1= \{1x: x\in
\{0,1\}^{p-1}\}$. Indeed, suppose that edge $(0x,1y) \not\in H$.
Then the shortest path in $H$ between $0x$ and $1y$ must be of
length at least $2^p+1$. This however violates the distortion
constraint. Now consider the set of points $A=\{e[2^{p-1}]: e =
(0x,1y, x,y \in \{0,1\}^{p-1}\}$. Clearly for any $a,b \in A$,
$d(a,b) \leq 3\cdot 2^{p-1}$ and $d(a,b) \geq 2\cdot 2^{p-1}$. The
claimed bound on the doubling dimension follows.
\end{proof}

Theorem~\ref{thm:graphlowerbound} follows.

\subsection*{Acknowledgments}

We thank James Lee for pointing out that a weaker version of
Theorem~\ref{thm:main1} could be inferred from Semmes' results. We also
thank Robi Krauthgamer and Ravishankar Krishnaswamy for discussions.


\end{document}